\documentclass[11pt]{article} 

\usepackage[utf8]{inputenc} 
\usepackage{authblk}
\usepackage[dvipsnames]{xcolor}

\usepackage{geometry} 
\usepackage{graphicx}
\graphicspath{ {images/} }
\usepackage{hyperref}

\usepackage{graphicx}
\usepackage{booktabs}
\usepackage{array}
\usepackage{paralist}
\usepackage{verbatim} 
\usepackage{subfig} 
\usepackage{amssymb}
\usepackage{amsmath}

\usepackage{fancyhdr} 
\usepackage[nottoc,notlof,notlot]{tocbibind} 
\usepackage[titles,subfigure]{tocloft}

\usepackage{amsthm}

\newtheorem{theorem}{Theorem}[section]
\newtheorem{corollary}[theorem] {Corollary}

\newtheorem{conjecture}[theorem] {Conjecture}
\newtheorem{lemma}[theorem]{Lemma}

\theoremstyle{definition}
\newtheorem{definition}{Definition}[section]

\newcommand{\poly}{\mathrm{poly}}

\newcommand{\C}{{\mathcal{C}}}

\newcommand{\ket}[1]{\left | #1 \right \rangle}

\def\pr{{\rm prob}}

\def\cd{{\cal D}}
\def\cp{{\cal P}}
\def\cb{{\cal B}}

\def\cf{{\cal F}}

\def\ca{{\cal A}}
\def\co{{\cal O}}
\def\ct{{\cal T}}
\def\cc{{\cal C}}

\begin{document}

\title{ 
Quantum advantage of unitary Clifford circuits\\ with magic state inputs}

\author{Mithuna Yoganathan, Richard Jozsa and Sergii Strelchuk\\
  \small\it DAMTP, Centre for Mathematical Sciences, University of Cambridge,\\ \small\it Wilberforce Road, Cambridge CB3 0WA, U.K.
  }

\date{}

\maketitle

\begin{abstract}

We study the computational power of unitary Clifford circuits with solely magic state inputs (CM circuits), supplemented by classical efficient computation. We show that CM circuits are hard to classically simulate up to multiplicative error (assuming PH non-collapse), and also up to additive error under plausible average-case hardness conjectures. Unlike other such known classes, a broad variety of possible conjectures apply. Along the way we give an extension of the Gottesman-Knill theorem that applies to universal computation, showing that for Clifford circuits with joint stabiliser and non-stabiliser inputs, the stabiliser part can be eliminated in favour of classical simulation, leaving a Clifford circuit on only the non-stabiliser part. 
Finally we discuss implementational advantages of CM circuits.

\end{abstract}

\section{Introduction}

A fundamental goal of quantum complexity theory is to prove that quantum computers cannot be efficiently simulated by classical computers. An approach to proving this was put forward by Bremner et al. \cite{Bremner2011}, showing that if a particular class of quantum circuits, so-called IQP circuits, could be efficiently classically simulated up to multiplicative error then the polynomial hierarchy (PH) would collapse. However on physical grounds it is more natural to consider classical simulations with additive or $l_1$ error. In this vein, Aaronson and Arkhipov \cite{Aaronson2011} showed that assuming the validity of two plausible complexity theoretic conjectures, the quantum process of boson sampling cannot be efficiently simulated up to additive error unless there is PH collapse. The conjectures are referred to as the anticoncentration conjecture and average-case hardness conjecture. Bremner, Montanaro and Shepherd \cite{Bremner2016} showed a similar result for IQP circuits, and furthermore they were able to prove the anticoncentration conjecture in their context. Since then, there have been further similar results for various classes \cite{Morimae2017, Bouland2017, Pashayan2017, Miller2017, Mann2017}.

In this paper we introduce a  subclass of quantum computing that we call Clifford Magic (CM), inspired by the PBC (Pauli Based Computing) model of Bravyi, Smith and Smolin \cite{Bravyi2016}, and establish a variety of its properties.The class CM comprises quantum circuits of unitary Clifford gates with fixed input $|A\rangle^{\otimes t}$ (for $t$ qubit lines) where $|A\rangle=\frac{1}{\sqrt{2}}(|0\rangle+e^{i \pi/4}|1\rangle)$ and with output given by final measurement of  some number of qubits in the computational basis. For computational applications we will also allow classical polynomial time computation for assistance before and after the Clifford circuit is run, in particular to determine the structure of  a CM process ${\cal C}_w$ for each computational input bit string $w$. If the Clifford gates could adaptively depend on further intermediate  measurements (not allowed here), the latter model would be universal for quantum computation, but our model appears to be weaker than universal. Our main result is to show that nevertheless, this class is hard to classically simulate up to additive error, given any one of a broad variety of average-case hardness conjectures. 

This result has been shown in the recent works \cite{Bouland2017} and \cite{Pashayan2017} (and our results were developed independently concurrently) but only for a single particular hardness conjecture. Furthermore both papers prove the anticoncentration conjecture by using the fact that random Clifford circuits form a $k$-design for suitable $k$. The idea of using $k$-designs to prove anticoncentration conjectures is explored in \cite{Hangleiter2017}. In this paper, we use a different approach. We show that this class, although unlikely to be universal, suffices to emulate the hardness of other classes of computations  already known to have the desired properties, thereby establishing hardness of CM simulation up to additive error, given any one of a number of inherited hardness conjectures.

Along the way we also establish  a generalised form of the Gottesman--Knill theorem viz. that any adaptive Clifford computation (now allowing intermediate measurements)  with input $\sigma \otimes \rho$, where $\sigma$ is a stabiliser state, can be simulated by an adaptive Clifford circuit on just $\rho$, with the help of polynomial time classical processing. This result amounts to a translation of the PBC model back into the circuit model, but has considerable conceptual interest in its own right, applying also to universal quantum computation. The standard Gottesman--Knill theorem \cite{Nielsen2010} is obtained in the case that the whole input is a stabiliser state and then the simulation can be done entirely classically. Thus for universal quantum computation represented in the model of adaptive Clifford circuits with magic state inputs \cite{Bravyi2004}, we can trade off part of the quantum processing for classical processing while compressing the quantum space requirement i.e. the number of qubits needed. 

Finally we will consider the feasibility of experimentally implementing CM circuits. This has become an increasingly relevant topic with the expected imminent availability of small quantum computers that may allow physical implementation of quantum algorithms unlikely to be simulatable even by the best classical computers \cite{Harrow2017}. We show that CM circuits have several properties that may make them advantageous for prospective experimental realisation in the near term. We show that in the measurement based computing model (MBQC), given the standard graph state, any CM circuit can be implemented without adaptions, and hence can be implemented in MBQC depth one. We also show that CM has good properties when it is made fault tolerant in both the circuit and MBQC models: while syndrome measurements must be performed, the associated correction operators need not be applied. Also, in MBQC given an initial state that can be created offline with high fidelity, CM can be implemented fault tolerantly with one further time step. 


\section{Preliminaries} \label{prelim}


$X$, $Y$ and $Z$ will denote the standard 1-qubit Pauli operations and ${\cal P}_n$ will denote the $n$-qubit Pauli group (generated by tensor products of the 1-qubit Pauli operations). $Z_i$ will denote the Pauli operation having $Z$ on the $i^{\rm th}$ line and $I$ on all other lines.
Pauli measurements for $P\in {\cal P}_n$ will have outcomes $\pm 1$. This applies to $Z_i$ measurements too, having outputs $\pm1 $ rather than bit values 0 and 1. We will state explicitly when the latter are used as output labels. A Pauli measurement $P$ is said to be dependent on Pauli measurements $Q_1, \ldots , Q_K$ if $P=\pm Q_1^{a_1}\ldots Q_K^{a_K}$ for some $a_1, \ldots , a_K\in \{ 0,1\}$.
$\ket{A}$ will denote the 1-qubit magic state $\ket{A}=\frac{1}{\sqrt{2}}(\ket{0}+e^{i \pi/4}\ket{1})$.

A stabiliser group $\mathcal{S}$ is a commuting subgroup of ${\cal P}_n$ that does not include $-\mathbb{I}$ . 
An $n$ qubit pure state $\ket{\psi}$ is a pure stabiliser state if it is stabilised by every element of a stabiliser group $\mathcal{S} $  (i.e. $S\ket{\psi}=\ket{\psi}$ for all $S\in \mathcal{S}$) that has $n$ independent generators (so then $\ket{\psi}$ is uniquely fixed by $\mathcal{S}$). More generally an $n$ qubit state $\rho$ is a mixed stabiliser state if it has the form 
\begin{equation}
\rho= \frac{1}{2^{n-s}}\prod \frac{\mathbb{I}+S_i}{2}.
\end{equation}
where $S_1, \ldots , S_s$ with $s\leq n$ are independent generators of a stabiliser group $\mathcal{S}$.
It is also stabilised by all the elements of $\mathcal{S}$ and may alternatively be described as the state produced by measuring the maximally mixed state with the (commuting) measurements $S_1$,...,$S_s$ and postselecting each  on outcome $+1$. 

Unitary Clifford circuits will always be assumed to be given as circuits of some chosen set of one and two qubit Clifford gates that suffice for any Clifford operation e.g. the Hadamard gate $H$, controlled NOT gate $CX$ and phase gate $S= {\rm diag}(1\,\,\, i)$. We will also consider circuits with intermediate $Z$ measurements and possibly adaptive choices of later gates, as formalised in the following definition.

\begin{definition}\label{adaptcircuit} An {\em adaptive quantum circuit} $C$ on $n$ qubits, with input state $\alpha$ and output distribution $P_C$ comprises the following ingredients. We have a specified sequence of steps (on the $n$-qubit state $\alpha$) of length poly$(n)$, with the following properties:\\
(i) each step is either a unitary gate or a non-destructive $Z$ basis measurement. Post-measurement states from intermediate measurements may be used as inputs to the next step.\\
(ii) each step is specified as a function of previous measurement outcomes by a classical (possibly randomised) poly$(n)$ time classical computation.\\ If no steps depend on previous measurement outcomes then the circuit is called {\em non-adaptive}, and if there are no intermediate measurements steps, then the circuit is called {\em unitary}.\\
The output distribution $P_C$ is the probability distribution of a specified set of measurements (called output measurements). Without loss of generality this may be taken to be the set of all measurements of the circuit $C$ and we often omit explicit mention of the output set.\, $\Box$

\end{definition}

We will use the non-Clifford $T$ gate defined by $T= {\rm diag}(1\,\,\, e^{i\pi/4})$.
It is well known that the $T$ gate can be implemented by the so-called $T$-gadget \cite{Nielsen2010}, using an extra ancilla qubit line (labelled $a$) in state $|A\rangle$  and adaptive Clifford operations:  to apply $T$ to a qubit line $k$ in a circuit, we first apply $CX_{ka}$ with the ancilla as target qubit, and then measure the ancilla qubit in the $Z$ basis giving outcome $+1$ or $-1$ (always with equal probability). Finally an $S$ correction is applied to the original qubit line if the outcome was $-1$. The ancilla qubit is never used again and may be discarded. The final result in every case is to apply $T$ to line $k$ up to overall phase. It will also be useful to note that we can implement the $T^\dagger$ gate using a similar gadget: we perform the $T$-gadget process as above but for the final adaptive correction we instead apply an $S^3$ correction if the outcome was $+1$.

Clifford operations with $T$ gates are universal for quantum computation. Using the $T$-gadget we see that any (universally general) circuit composed of Clifford gates and a number $t$ of $T$ gates can be rewritten as an adaptive circuit of only Clifford gates (and intermediate $Z$ basis measurements) with the addition of $t$ additional ancilla qubit lines initialised in state $|A\rangle^{\otimes t}$. 

Finally, we define a notion of weak simulation of one quantum process by another, that we will use in this work. 
\begin{definition}\label{weaksimdef} We say that a circuit $C$ (on $n$ qubits, with input state $\alpha$, and output distribution $P_C$) can be {\em weakly simulated} by a circuit $\tilde{C}$ (on $m$ qubits, with input state $\beta$, and output distribution $P_{\tilde{C}}$) if\\
(i) a description of the circuit $\tilde{C}$ may be given by a classical poly$(n)$ time (possibly randomised) translation from a description of $C$, and\\
(ii) a sample of the distribution $P_C$ can be produced from a sample of $P_{\tilde{C}}$ together with poly$(n)$ time classical (randomised) computation.\, $\Box$

\end{definition}

(More precisely, in the above definitions the poly$(n)$ bounds refer to a situation in which we are considering a uniform family of circuits depending on an associated parameter $n\in \mathbb{N}$, which will be clear from the context when needed.)

%
%

\section{Extending the Gottesman--Knill theorem}


We begin by establishing an extended form of the Gottesman--Knill theorem that will be used later in our development of CM circuits.

The standard form of the Gottesman-Knill theorem asserts that any adaptive Clifford circuit with stabiliser state input may be classically efficiently weakly simulated \cite{Gottesman2008,Jozsa2013}. As noted above, universal quantum computation can be performed using adaptive Clifford circuits which include additional (non-stabiliser) $\ket{A}$ state ancilla inputs, motivating the consideration of Clifford circuits on such more general inputs. In our extension of the Gottesman-Knill theorem we consider adaptive Clifford circuits but now allow the input to have a non-stabiliser part. We show that it may be weakly simulated by a hybrid classical-quantum process whose quantum part (obtained by an efficient classical reduction from the description of the original circuit) is an adaptive Clifford circuit acting now only on the non-stabiliser part of the original input, thereby relegating the stabiliser-input part of the original computation into efficient classical computation instead. In the special case where the initial input is fully a stabiliser state, we recover the standard Gottesman--Knill theorem, as our hybrid process then has no residual quantum part. This is stated formally as follows:

\begin{theorem} \label{EGK} (Extended Gottesman--Knill Theorem)
Let $\mathcal{C}$ be any adaptive Clifford circuit with input state  $\sigma\otimes\rho$, where $\sigma$ is a stabiliser state of $n$ qubits and $\rho$ is an arbitrary state of $t$ qubits, and the output is given by measurement of any specified qubit lines. (Usually we will also have $t=O({\rm poly}(n))$). Then\\
(i)  $\mathcal{C}$ can be weakly simulated by an adaptive Clifford circuit $\mathcal{C}^*$ on $t$ qubits with input $\rho$, assisted by poly$(n+t)$-time classical computation, and with $\mathcal{C}^*$ having at most $t$ (intermediate or final) measurements;\\
(ii) if  $\mathcal{C}$ is non-adaptive then $\mathcal{C}^*$ may be taken to be unitary (with $Z$ basis measurements only for outputs at the end). \\
(iii) If some $Z$ measurements in $C$ are to be postselected to outcome $+1$, this circuit can be weakly simulated by a circuit $\mathcal{C}^*$ as in case (i), where some of the $Z$ measurements are postselected to outcome $+1$. 
\end{theorem}

The proof of the Extended Gottesman--Knill Theorem will be given in Subsection \ref{egkproof} below. It rests on the so-called Pauli based model of computation (PBC) introduced by Bravyi, Smith, and Smolin in \cite{Bravyi2016}. Before the proof of Theorem \ref{EGK} we will in Subsection \ref{pbcmodel}, give an account of (a slightly generalised version of) the PBC formalism and its main features that we will use.

The Extended Gottesman-Knill theorem will be used in this paper to show that certain quantum circuits can be simulated by CM circuits (cf Section \ref{cm}). However, we expect that the theorem will be of independent interest, for example for considerations of compiling quantum circuits with as few qubits as possible. Indeed starting with the circuit model of quantum computation we may represent any circuit as a circuit of Clifford gates and $T$ gates, and then use $T$-gadgets to implement the $T$ gates, resulting in an adaptive Clifford circuit. Implementing the circuit this way allows for error correction using stabiliser codes \cite{Nielsen2010}, but it also increases the number of qubits. Given the high practical cost of adding extra qubits, one naturally strives to minimise their number in near term devices. The Extended Gottesman--Knill theorem provides a way to remove all qubits originally in a stabiliser state, as well as any stabiliser ancillas. The resulting circuit is also an adaptive Clifford circuit, now having at most $t$ measurements. This is summarised in Figure 1.

\begin{figure}[h!] 
\label{compression}
    \begin{center}
    \includegraphics[trim=120 650 0 50,width=\textwidth]{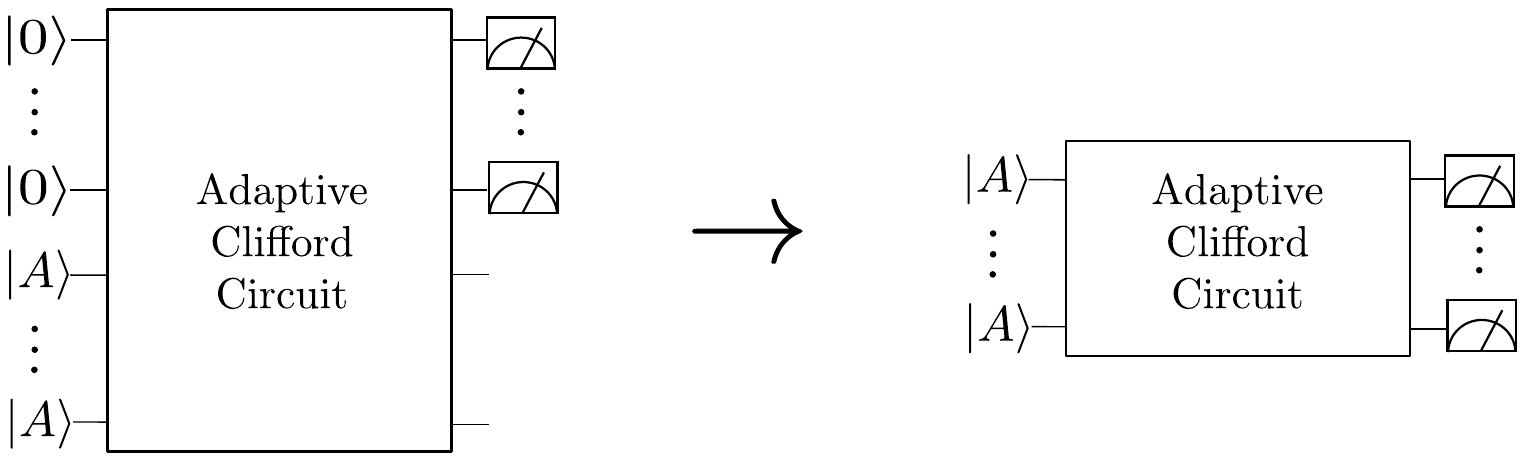}
    \caption{The Extended Gottesman Knill theorem (Theorem \ref{EGK}) allows us to take a universal quantum circuit expressed as a Clifford circuit with $T$-gadgets and compress it using only a classical polynomial time overhead. This compression removes all input state components that are stabilisers and the resulting circuit is an adaptive Clifford circuit with a number of (intermediate and final) measurements at most equal to the number of lines in the compressed circuit.}
    \end{center}
   
\end{figure}

In \cite{Aaronson2004} and \cite{Bravyi2016a} a different kind of extension of the Gottesman-Knill theorem is developed. It is shown that a circuit on $n$ qubit lines with stabiliser input and $t$  $T$ gates, can be classically simulated in time exponential in $t$ and polynomial in $n$. This reduces to the original Gottesman--Knill theorem when $t=0$. 
Our Extended Gottesman Knill theorem provides an alternative proof of this fact: using Theorem \ref{EGK} any such computation (after replacing $T$ gates by $T$-gadgets) can be compressed to a quantum computation on $t$ qubits, and this can be and then be classically simulated in time exponential in $t$.

\subsection{The Pauli based model of computation (PBC)}\label{pbcmodel}

\begin{definition}\label{pbcdef} (PBC circuits and the Pauli based computing model)\\
(i) A {\em PBC circuit} $C$ on $t$ qubits with any input state $\rho$, is a sequence $C$  of pairwise commuting and independent Pauli measurements $P_1, \ldots ,P_s$ from $ \cp_t$ (applied sequentially to $\rho$ with each post-measurement state being available for the next measurement). The choice of each $P_i$ can generally adaptively depend on previous measurement outcomes. If no $P_i$ depends on previous measurement outcomes then the PBC circuit is called non-adaptive.\\
(ii) For computational applications ({\em the PBC model of computing}) we will use a uniform family $\{ C_w: w\in \cb \}$ of PBC circuits on $t= {\rm poly}(n)$ qubits where $n$ is the length of the bit string $w$, and furthermore, each $C_w$ is required to have the input state $\rho =\ket{A}^{\otimes t}$. The result of the computation is given by a specified poly$(n)$ time (randomised) classical computation on $w$ together with the measurement outcomes of the circuit $C_w$.
\,$\Box$ 

\end{definition}

\begin{theorem}\label{pbcthm}  (adapted from Ref \cite{Bravyi2016}). Let $C$ be any (generally adaptive) quantum circuit on $n+t$ qubits with input state $\alpha = \sigma\otimes \rho$ where $\sigma$ is a stabiliser state of $n$ qubits and $\rho$ is any state of $t$ qubits. Suppose also that the unitary steps of $C$ are all Clifford gates. Then:\\
(i) $C$ may be weakly simulated by a (generally adaptive) PBC circuit $\tilde{P}_1,\ldots ,\tilde{P}_s$ on $t$ qubits with input state $\rho$, and with $s\leq t$ steps.\\ 
(ii) If $C$ is non-adaptive (with final $Z$ basis measurement outputs) then the PBC circuit $\tilde{P}_1,\ldots ,\tilde{P}_s$ in (i) can also be chosen to be non-adaptive.\\
(iii) If some $Z$ measurements in $C$ are to be postselected to outcome $+1$, then this circuit can be weakly simulated by a PBC circuit in which some of the Pauli measurements are postselected to outcome $+1$. 
\,
$\Box$
\end{theorem}


We give the proof in full (following the method of \cite{Bravyi2016} and extending the latter for clauses (ii) and (iii) above) dividing it into labelled sections. We begin with two supporting lemmas.

\begin{lemma}\label{lemmap} \cite{Bravyi2016} Let $P, Q\in \cp_n$ be anti-commuting Pauli operations and let $\ket{\psi}$ be an eigenstate of $P$ with $P\ket{\psi}=\lambda_P\ket{\psi}$, $\lambda_P=\pm 1$. Then:\\
(i) Measurement of $Q$ on $\ket{\psi}$ gives result $\lambda_Q=\pm 1$ with equal probabilities half.\\
(ii) The operator $V(\lambda_P,\lambda_Q)=(\lambda_PP+\lambda_QQ)/\sqrt{2}$ is always a unitary Clifford operation.\\
(iii) $V(\lambda_P,\lambda_Q)\ket{\psi}$ is the normalised projection of $\ket{\psi}$ onto the $\lambda_Q$-eigenspace of $Q$.\\
{\em Hence measurement of $Q$ on $\ket{\psi}$ is equivalent to classically choosing (offline) a uniformly random $\lambda \in \{ -1,+1\}$ and applying the Clifford unitary $V(\lambda_P,\lambda)$ to $\ket{\psi}$.}
\end{lemma}

\begin{proof} We have $\ket{\psi}=\lambda_PP\ket{\psi}$.\\ 
For (i) we have 
$ \mbox{Prob\,$(Q$ measurement gives $\pm 1)$} = \left|\left| \frac{1}{2}(I\pm Q)\ket{\psi}\right|\right|^2.   $
Replacing $\ket{\psi}$ by $\lambda_PP\ket{\psi}$, and using the fact that $PQ=-QP$ and that $P$ is unitary, we readily see that the two probabilities are equal.\\
For (ii), using $P^2=Q^2=I$ and $PQ=-QP$ we can check directly that $V(\lambda_P,\lambda_Q)V(\lambda_P,\lambda_Q)^\dagger=I$. Similarly for any Pauli $R$, for each of the four possible combinations of $R$ commuting or anti-commuting with $P$ and $Q$, we can check directly that  $V(\lambda_P,\lambda_Q)\,R\,V(\lambda_P,\lambda_Q)^\dagger$ is a Pauli operation (being just a suitable product of $P$, $Q$ and $R$ in each case).\\
For (iii) the normalised post-measurement state after outcome $\lambda$ is 
\[ \frac{(I+\lambda Q)}{\sqrt{2}}\ket{\psi}=
\frac{(\lambda_PP+\lambda Q)}{\sqrt{2}}\ket{\psi}=V(\lambda_P,\lambda)\ket{\psi}.\hspace{1cm} \]

\end{proof}
We will also use the following fact which is easily checked.

\begin{lemma}\label{twiddle} For any $P=\pm A_1\otimes \ldots \otimes  A_n\otimes B_1\otimes \ldots \otimes B_t \in \cp_{n+t}$ with all $A_i$'s and $B_j$'s being $X,Y,Z$ or $I$, write $\tilde P = \pm B_1\otimes \ldots \otimes B_t \in \cp_{t}$ (with same overall sign as $P$). If $P$ commutes with $Z_1, \ldots ,Z_n \in \cp_{n+t}$ then each $A_i$ is either $Z$ or $I$. If for all $i$, each $A_i$ is either $I$ or $Z$, then for any $t$-qubit state $\ket{\psi}$, the measurement of $P$ on $\ket{0}^{\otimes n}\ket{\psi}$, and the measurement of $\tilde P$ on $\ket{\psi}$, give the same output distributions and corresponding post-measurement states of the form $\ket{0}^{\otimes n}\ket{\psi'}$ and $\ket{\psi'}$ respectively, with the same $t$-qubit states $\ket{\psi'}$.

\end{lemma}

\noindent{\bf Proof of Theorem \ref{pbcthm}}

Let $\mathcal{C}$ be any adaptive circuit whose steps are either unitary Clifford gates or $Z$ measurements, with $K$ measurements in total.  For clarity, we will give the proof for the case where $\sigma$ is the pure state $\ket{0}^{\otimes n}$. The general case of arbitrary (mixed) stabiliser state $\sigma$ is proved similarly by just replacing $Z_1,\ldots , Z_n$ in (b) below by a set of generators $S_1, \ldots S_r$ $(r\leq n)$ of the stabiliser group defining $\sigma$.

{\bf (a)} Starting with the rightmost Clifford gate and working successively to the left, we commute each gate out to the end of the circuit beyond the last measurement. As a result each $Z$ measurement will become conjugated into a Pauli measurement $P_i\in \cp_{n+t}$ which may be efficiently determined. Unitary gates applied after the measurements have no effect on the outcomes so we delete them, and we are left with a sequence  $P_1, P_2, \ldots , P_K$ of (generally adaptive) Pauli measurements (where $s$ is the number of $Z$ measurements in $\mathcal{C}$), acting on input  state
$\ket{0}^{\otimes n}\otimes \rho$.\\
Remark on (a): we could instead commute out the Clifford gates in sections, interleaved with the process to be described in (c) below, as follows. As we consider each successive measurement $Q_i$ of the original circuit in turn (working from the leftmost one) we commute only the Clifford gates on the left of $Q_i$ to the right of it, and staying to the left of the next measurement, to obtain $P_i$ as above, and then apply (c) to $P_i$. All gates are thus eventually commuted out beyond the last measurement as we consider each measurement in turn. This commuting process interleaved with (c) has the advantage that for adaptive gates (depending on previous measurement outcomes) the identity of the gate is always fixed before it is commuted to the right, and we never need to carry forward any variables of adaptation.

{\bf (b)} Next we prefix the sequence in (a) with ``dummy'' $Z$ measurements for each of the first $n$ lines obtaining the list
\[ {\rm (LIST):}\hspace{5mm} Z_1, Z_2, \ldots , Z_n, P_1, P_2, \ldots , P_K. \]   
This has no effect as the input is $\ket{0}$ on each of these lines (and the $Z$ measurements all give result $+1$ with certainty).

{\bf (c)} We now define our PBC process. We have a $t$-qubit register initially in state $\rho$. Looking at (LIST) in (b) we work successively through the $P_j$'s starting with $P_1$(not the dummy $Z$'s). For each $P_j$:\\
(i) If $P_j$ is dependent on measurements already performed (which may be efficiently determined  \cite{Nielsen2010}), delete $P_j$ from (LIST) and just calculate its outcome from previous recorded measurement results. Move to the next measurement in (LIST).\\
(ii) If $P_j$ commutes with all measurements to the left in (LIST) (including the dummy $Z$'s too), measure $\tilde P_j$ (as in Lemma \ref{twiddle}) on the register and record its value $\lambda_{P_j}$. Then move to the next measurement in (LIST).\\
(iii) If $P_j$ anticommutes with some measurement $N$ (possibly a dummy $Z$) on the left (which had outcome $\lambda_N$),  classically randomly choose $\lambda_{P_j}\in \{ +1,-1\}$ and record it. Then delete $P_j$ from (LIST) and replace it by the unitary Clifford $V(\lambda_N,\lambda_{P_j})$ (as in Lemma \ref{lemmap}). Then update (LIST) by commuting out $V(\lambda_N,\lambda_{P_j})$ to the right. By Lemma \ref{lemmap} this process simulates the $P_j$ measurement and its post-measurement state for subsequent measurements. Then move to the next measurement in (LIST).

It is clear that when we have treated all $P_j$'s in (LIST) we will have performed a list of $s\leq K$ measurements on the $t$-qubit register, which are independent and commuting Pauli measurements (the only quantum action on the register occurring in (ii)), and this process is assisted by efficient randomised classical computation. Since the measurements are all independent and commuting, we must have $s\leq t$.

Independently of actually implementing the measurements on the quantum register, the process described in (c) above provides an efficient classical (generally randomised) procedure which, given a sequence of measurement outcomes $m_1, \ldots , m_l$ up to any stage $l$, determines the next quantum measurement that's guaranteed to be independent of all previous measurements and commuting with them i.e. a bonafide PBC circuit. This completes the proof of Theorem \ref{pbcthm}(i).

{\bf (d)} We now prove Theorem \ref{pbcthm}(ii). If $\mathcal{C}$ is non-adaptive then we may assume without loss of generality that it is a unitary circuit $U$ followed by final measurements $Z_{i_1}, \ldots , Z_{i_s}$ on specified qubit lines $i_1, \ldots , i_s$ \cite{Jozsa2013}. Then in (b) we will obtain the non-adaptive list $ Z_1, Z_2, \ldots , Z_n, P_1, P_2, \ldots , P_s. $
Here $P_k=UZ_{i_k}U^\dagger$ for $k=1, \ldots , s$, which are commuting and independent. However some may anticommute with an initial dummy $Z$ measurement. Then following the process of (c)(iii) (with $P_j$ and $N$ as in (c) above), $N$ must be one of the dummy $Z$'s, whose measurement outcome $\lambda_N=+1$ is deterministic. Thus the unitary gate $V(\lambda_{P_j}, \lambda_N)$ involves no adaptations, and the sequence remains non-adaptive after $V(\lambda_{P_j}, \lambda_N)$ is commuted out to the end (although it depends on the classical random choice of $\lambda_{P_j}$ that can have been chosen a priori). Continuing in this way, we note that if any subsequent updated operator $M$ anticommutes with any earlier operator $N$, then $M$ must always anticommute with one of the dummy $Z$'s too. This is because at any iteration stage, the operators after the dummy $Z$'s are given by initial $P_i$'s conjugated some number of times by operators $V$ that are always in the algebra generated by the $P_k$'s and dummy $Z$'s (i.e. the successive $V$'s that have been commuted out). Thus if $M$ commuted with all the dummy $Z$'s, it must also commute with all preceding operators $N$ (recalling that the $P_k$'s were all commuting).

Now by choosing an anticommuting $N$ to always be a dummy $Z$, $\lambda_N$ will always be $+1$ and no adaptation is ever introduced by (c)(iii) so, since the initial list of $P_i$'s was non-adaptive, the final PBC process will be non-adaptive too. This proves Theorem \ref{pbcthm}(ii). 

{\bf (e)} Finally we prove Theorem \ref{pbcthm}(iii). In the case of postselection we proceed with all the steps as above as though there was no postselection, except (c)(iii). Suppose that the measurement $P_j$ in that step is postselected to outcome $+1$. In that case, do not randomly choose $\lambda_{P_j}$, but set it to $\lambda_{P_j}=1$. Replacing $P_j$ with $V(\lambda_N,1)$ will produce the same post measurement state as postselecting $P_j$ on outcome $+1$.  If a dependent measurement's determined outcome (as in (c)(i)) is inconsistent with an imposed postselection at that stage, then this indicates that the postselection requirement of the original circuit had probability zero. This results in a PBC process, some of whose measurements (arising from (c)(ii)) may still be postselected,
 completing the proof of Theorem \ref{pbcthm}(iii). \,$\Box$

\subsection{Proof of the extended Gottesman-Knill theorem}\label{egkproof}

A PBC circuit with general input state $\rho$ is similar to an adaptive Clifford circuit albeit with no unitary gate steps, except that the measurements are general Pauli measurements rather than just elementary $Z$ measurements. Correspondingly our extended Gottesman-Knill Theorem \ref{EGK} is obtained as a translation of Theorem \ref{pbcthm} into a standard circuit form.


\noindent {\bf Proof of Theorem \ref{EGK}} \\ 
According to Theorem \ref{pbcthm}(i), $\mathcal{C}$ can be weakly simulated by a PBC circuit of Pauli measurements $\tilde P_1,...,\tilde P_s$ on input state $\rho$, and we just need to translate this back into an adaptive Clifford circuit with only $Z$ basis measurements. This follows immediately by applying lemma \ref{umap} below to each $\tilde P_i$ separately, expressing it as $\tilde P_i=U_i^\dagger Z_k U_i$ for unitary Clifford operations $U_i$ and any choice of line $k$ (which could even be independent of $i$), thus establishing (i) and (iii). 

Note that the Lemma cannot be applied to all $\tilde P_i$ simultaneously (giving a single $U$) since although pairwise commuting and independent, they are generally adaptively determined and not fixed a priori. 
However if $\mathcal{C}$ is non-adaptive then according to Theorem \ref{pbcthm}(ii), the sequence $\tilde P_1,...,\tilde P_s$ can be chosen to be non-adaptive. Lemma \ref{umap} can then be applied to the whole list to give a single $U$ with $U^\dagger Z_k U= \tilde P_k$ for $k=1,\ldots ,s$. The circuit $\mathcal{C}^*$ is then just the unitary Clifford $U$ (as unitaries after the $Z$ measurements have no effect and can be deleted), thus establishing (ii).

\begin{lemma} \label{umap}
Let $\{P_1,...,P_m\}$ be any set of independent and pairwise commuting Pauli operations on $n$ qubits (so $m\leq n$).  
Then there is a unitary Clifford operation $U$ such that $U^\dagger Z_k U= P_k$ for $k=1,\ldots ,m$. Furthermore a circuit of basic Clifford gates of depth $O(n^2/log(n))$ implementing $U$ may be determined in classical poly$(n)$ time.  \end{lemma}

\begin{proof}
We first extend the set $\{P_1,...,P_m\}$ to a maximally sized set $\{P_1,...,P_n\}$ of independent  pairwise commuting Pauli operations.  This extension is not unique, but see Section 7.9 of \cite{preskill} for an efficient method of extension. Using similar techniques we also find generators of the `destabiliser group' $\{D_1,...,D_n\}$ (defined in \cite{Aaronson2004, yoder}). Then there is a unique (up to phase) Clifford $V$ such that $V Z_i V^\dagger= P_i$ and $V X_i V^\dagger= D_i$ for $i=1,\ldots , n$. An $O(n^2/log(n))$ circuit implementing $V$ may be determined in classical poly$(n)$ time by  the construction of Theorem 8 in \cite{Aaronson2004}. Finally take $U=V^\dagger$.
\end{proof}


%
%

\section{Clifford magic (CM) circuits}\label{cm}

We introduce a class of quantum processes that we call ``Clifford Magic", written CM.

\begin{definition}
A CM circuit on $t$ qubits is a unitary Clifford circuit which has input state $\ket{A}^{\otimes t}$, and output given by the result of measuring $r$ specified qubits (the output register $\co$) in the $Z$ basis (and intermediate measurements are not allowed). A postselected CM circuit is a CM circuit with an additional register $\cp$ of $s$ qubits (called the postselection register) disjoint from $\co$, which is also measured at the end.\, $\Box$
\end{definition}

Our motivation for introducing and studying CM circuits is twofold. The first reason, discussed in Subsection \ref{classsimresults}, relates CM processes to known classical simulation results. In particular, we show that the class of CM circuits is equivalent to a class of quantum circuits likely to have supra-classical power while also being weaker than BQP. 
Our second motivation, discussed in Subsection \ref{exper}, is that CM circuits are a promising candidate for experimentally verifying quantum advantage. Unlike other quantum supremacy proposals, small amounts of error correction can be readily included with modest overheads. Furthermore,  adding adaptive measurements to CM processes makes the class universal while also providing an economy in the number of qubits needed, as described previously in Figure 1. In this way CM circuits may be viewed as a practicable stepping stone towards an implementation of universal quantum computation. 

\subsection{Relation between CM and known classical simulation results} \label{classsimresults}

Consider circuits of the form shown in Figure 2. The circuits on the left comprise unitary Clifford gates with input $\ket{0}^{\otimes n}\ket{A}^{\otimes \poly(n)}$ and one line being measured for the output. Such circuits are known to be classically simulatable \cite{Jozsa2013}. On the other hand, if intermediate $Z$ measurements are allowed together with adaptations, the circuits can perform $T$-gadgets making them universal for BQP computations, as shown on the right.

\begin{figure}[h!] 
\label{f}
    \begin{center}
    \includegraphics[trim=25 600 -10 50,clip,width=\textwidth]{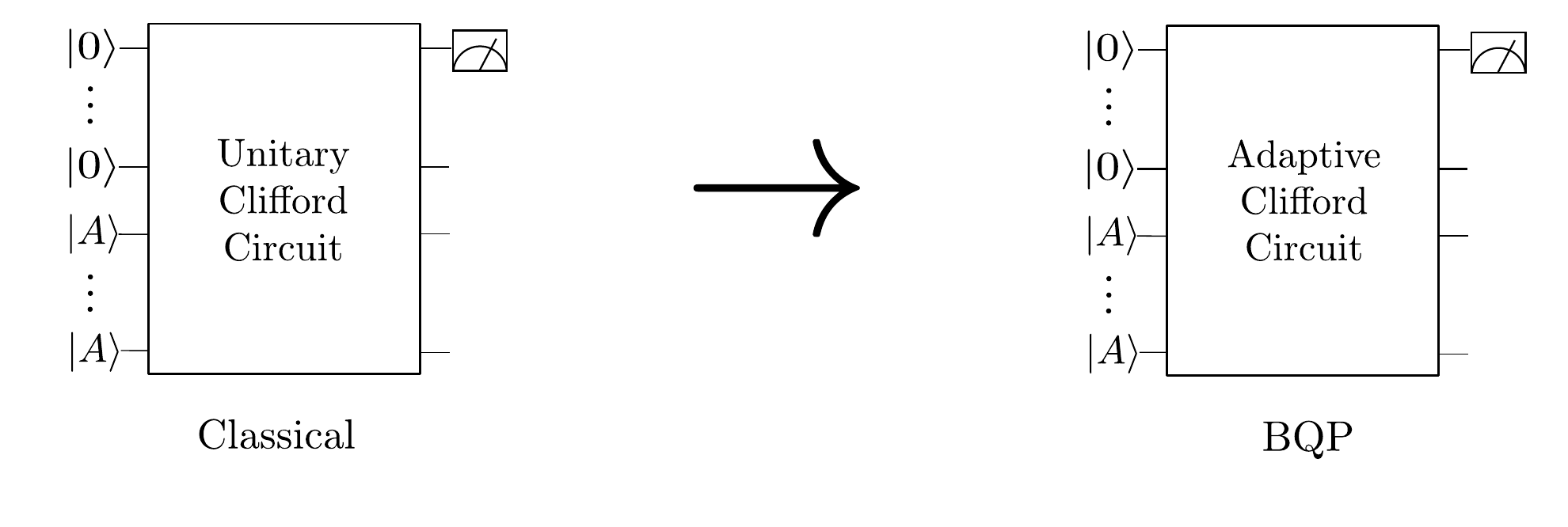}
    \caption{The circuits on the left have magic states as well as stabiliser inputs. However, if a unitary Clifford circuit is applied and only one line is measured, it is classically simulatable. On the other hand, if intermediate $Z$ measurements are included and the circuit is allowed to adaptively depend on measurement outcomes, then the circuit can perform any BQP computation.}
    \end{center}
   
\end{figure}

 Consider now the family of all Clifford circuits with input $\ket{0}^{\otimes n}\ket{A}^{\otimes \poly(n)}$ and one line being measured for the final output, and allowing intermediate measurements. Let ${\cal M}_I$ denote the set of intermediate measurement results obtained. Then we can consider ${\cal M}_I$ being used in one of the following three ways:\\
(A)  Discarding  ${\cal M}_I$, and not using it in any way (either for output or for adaptations).\\
(B)  Retaining ${\cal M}_I$ as part of the output (but not used otherwise). \\
(C) Using ${\cal M}_I$ as it emerges for subsequent adaptation in the course of the process, as well as giving ${\cal M}_I$ as part of the output.\\ 
Circuits of the form (C) can perform any BQP computation, but those of the form (A) are classically simulatable \cite{Jozsa2013}. 
Case (B) is not expected to have the full power of BQP. But furthermore, using the methods of \cite{Jozsa2013} (cf especially Theorems 6 and 7 therein, and under plausible complexity conjectures) case (B) is also not classically simulatable exactly (in either the strong or weak sense). In this work (cf Section \ref{hard}) we will show that additionally, it is also not classically simulatable up to multiplicative or additive error either (under plausible conjectures).


Case (B) is clearly intermediate between (A) and (C). Indeed (C) allows the extra capability over (B) of adaptation, and compared to (A), retaining ${\cal M}_I$ in (B) gives more information about the final state which in (A) would be assigned as the probabilistic mixture of all post-measurement states arising from all the possible outcome values for ${\cal M}_I$.


The class of CM circuits is clearly a subset of the class of circuits in case (B) viz. those with no $\ket{0}$ part in the input and all measurements being performed only at the end. However, the CM subset is in fact equivalent to the full class in (B): every circuit in the latter can be weakly simulated by a CM circuit, as follows by an application of the Extended Gottesman--Knill theorem. As the intermediate measurements in case (B) are not adaptive, Theorem \ref{EGK}(ii) tells us that the resulting compressed circuit is a CM circuit. 

In this sense the computational power of the class of CM circuits relates directly to the power of retaining intermediate measurements in a Clifford circuit. We prove in Section \ref{hard} that CM circuits cannot be classically simulated (up to multiplicative or additive error) under plausible conjectures, showing that the mere retention of intermediate measurement results as above, can be regarded as a kind of ``quantum resource'', elevating the classically simulatable case (A) to supra-classical computing power in (B).

\subsection{Experimental advantages of CM circuits} \label{exper}

CM circuits offer several advantages for fault tolerant implementation and for implementation in the MBQC model, inherited in part from such benefits for Clifford circuits. 

\subsubsection{Fault tolerance for CM circuits}


In the circuit model, fault tolerance is often achieved by replacing $T$ gates by $T$ gadgets, with magic state distillation being used to create high fidelity $|A\rangle$ states offline \cite{Bravyi2004}. However, as $T$ gadgets include adaption, the circuit cannot be fully created in advance, and instead part of the circuit must be created in real time. These potentially increase the required coherence times. CM do not require these kinds of adaptions, even when made fault tolerant using a stabiliser code. 

Syndrome measurements and their associated correction operations may appear to introduce further adaptations into the circuit, but these can in fact be avoided. Indeed these corrections are Pauli operations, and can always be commuted past Clifford unitaries and (Pauli) syndrome measurements, since the Pauli measurements, at most, swap sign when conjugated by the Pauli corrections. Then the Pauli corrections can be accounted for after the quantum computation is completed via simple classical processing of the measurement outcomes. 

A further benefit of CM circuits being Clifford circuits is that any such circuit on $t$ qubit lines can be expressed as a circuit of depth bounded by $O(t^2/\log t)$ \cite{Aaronson2004}, again providing potential benefits for shorter coherence times in implementation. 
 
\subsubsection{CM circuits in the MBQC model}

In our discussion below we will assume the following standard form of MBQC (cf for example \cite{Danos2007}). The starting resource state is the standard cluster state. $CZ$ operations in circuits are implemented by exploiting $CZ$'s that were used in the construction of the cluster state. 1-qubit measurements applied to the cluster state are either $Z$ measurements or else $M(\alpha)$ measurements in the basis  $\{|\pm_{\alpha}\rangle\}$, where $|\pm_{\alpha}\rangle=1/\sqrt{2} (|0\rangle \pm e^{-i \alpha} |1\rangle)$. The latter provide implementation of 1-qubit gates  $J(\alpha)=H (|0\rangle \langle 0|+ e^{i \alpha} |1\rangle \langle 1|)$, appearing as $X^sJ(\alpha)$ where $s=0,1$ is the measurement outcome and $X^s$ is the associated byproduct operator. The $J(\alpha)$ gates together with $CZ$ provide a universal set.

\begin{theorem}
A CM circuit $\mathcal{C}$ including preparation of its input $| A \rangle^{\otimes t}$, can be implemented in the MBQC model in depth 1. 
\end{theorem}

\begin{proof}
Note first that  $|A\rangle = H J(\pi/4) |+\rangle$. Thus $\mathcal{C}$ may be viewed as having input $|+\rangle$ on all lines, followed by a round of $J(\pi /4)$ gates, followed by Clifford gates (comprising a round of $H$ gates followed by the gates of $\mathcal{C}$). Hence for MBQC implementation the measurement pattern comprises a line of $M(\pi /4)$ measurements laid out next to  implementations of Clifford gates. The $X^s$ byproducts of the $M(\pi /4)$ measurements can be commuted over the Clifford gates to the end, without incurring any adaptations. Similarly it is well known \cite{RBB2003} that Clifford circuits can be implemented without adaptation to the byproduct operators that arise. Hence the entire measurement pattern is non-adaptive and can be implemented in depth 1.
\end{proof}

Miller et al. \cite{Miller2017} also propose a scheme for quantum supremacy without error correction that is depth $1$ in MBQC, based on use of MBQC to simulate IQP circuits. Their scheme requires a nonstandard resource state that may not be simple to prepare, whereas our proposal uses the standard cluster state, which is a stabiliser state, as the resource. Furthermore our scheme can be made fault tolerant as follows. 
\begin{theorem}

A CM circuit $\mathcal{C}$ can be implemented fault tolerantly in the MBQC model in depth 1, given a particular initial resource state that can be created offline with high fidelity. 
\end{theorem}

\begin{proof}

For simplicity, we will consider a fault tolerance scheme using the 7-qubit Steane code. The initial resource  state can be created as follows. Create an encoded magic state $\tilde{\ket{A}}^{\otimes t}$. Create the other parts of the encoded graph state by making the encoded states $|\tilde{+}\rangle$ and using the encoded version of $CZ$. The usual syndrome measurements and corrections are required during this process. Inclusion of  $\tilde{\ket{A}}^{\otimes t}$ into the resource state allows us to avoid a later need for implementing encoded $M(\pi /4)$ measurements fault tolerantly, and our CM circuit is a circuit of only Clifford gates. Now we have $H=J(0)$ and $S=HJ(\pi/2)$, with $M(0)$ and $M(\pi/2)$ being $X$ and $Y$ measurements respectively. Thus in MBQC, Clifford gates are implemented using only Pauli measurements, and in our encoded setup we need to apply their corresponding fault tolerant encoded versions. These are transversal. Furthermore, syndrome measurements can be carried out using the usual fault tolerant construction in terms of Clifford operations and ancillas. These Clifford gates themselves can be implemented using MBQC using ancillas. All these ancillas are included in the initial state. Hence every physical operation applied to the initial state is a $1$ qubit Pauli measurement. Then, as before, Pauli errors can be corrected via classical post processing, and so the circuit is depth $1$. 
\end{proof}

%
%

\section{Hardness of classical simulation of CM circuits} \label{hard}

We now establish lower bounds on the complexity of classical simulation of CM circuits, allowing either multiplicative or additive errors in the simulation. The scenario of additive error is generally regarded as a reasonable model of what is feasible to physically implement in practice.

A distribution $q(x)$ is an $\epsilon$-additive approximation of a distribution $p(x)$ if
\begin{equation} \label{defadd}
\sum_x |p(x)-q(x)| \leq \epsilon.
\end{equation}
A number $Y$ is an $\epsilon$-multiplicative approximation of a number $X$ if $|X-Y|\leq\epsilon X$. A distribution $q(x)$ is an $\epsilon$-multiplicative approximation of a distribution $p(x)$ if for each $x$, $q(x)$ is an $\epsilon$-multiplicative approximation of $p(x)$. Thus clearly $\epsilon$-multiplicative approximation of distributions implies $\epsilon$-additive approximation.

\subsection{Hardness of classical simulation of CM with multiplicative error} \label{clsimmulti}
Although (uniform families of) CM circuits themselves are not likely to be universal for quantum computation, we first establish that postselected CM circuits suffice as a quantum resource for postselected universal quantum computation. Using the arguments of Ref \cite{Bremner2011}, this is enough to establish that the class cannot be classically simulated to multiplicative error without causing the Polynomial Hierarchy (PH) to collapse. 




\begin{theorem} \label{postselection}
Any postselected poly-sized unitary quantum circuit $\mathcal{C}$ on $n$ qubits (with final $Z$ measurements) can be weakly simulated by a postselected poly-sized CM circuit on poly$(n)$ qubits. 
\end{theorem}

\begin{proof}

We may suppose without loss of generality that $\mathcal{C}$ has the following form: the input state is $|0\rangle^{\otimes n}$, followed by Clifford and $T$ gates, and finally some number of lines is measured in the $Z$ basis. Of these, some are postselected to outcome $k=+1$. To begin, we replace each $T$ gate with a $T$-gadget where the gadget measurement is postselected to outcome $+1$ so the correction $S$ is not required. As no other part of the circuit acts on this ancilla line again this measurement can be performed at the end of the circuit. The resulting circuit $\tilde{\mathcal{C}}$ then has input $|0\rangle^{\otimes n} |A\rangle^{\otimes t}$, which is acted on by a Clifford unitary $U$ followed by $Z$ measurements, some of which are postselected. The proof is now completed in either one of two possible ways, labelled (a) and (b), as follows:\\
(a) Theorem \ref{EGK}(ii) and (iii) can then be used to provide an algorithm for simulating the above circuit $\tilde{\mathcal{C}}$ by a postselected CM circuit. \\
(b) We start with the state $\ket{A}^{\otimes (n+t)}$ and first convert it to $|0\rangle^{\otimes n} |A\rangle^{\otimes t}$. This is achieved by applying a $T$-gadget postselected to outcome $-1$ (thus implementing a $T^\dagger$ gate), and then $H$, to each of the first $n$ qubits, and then we apply the Clifford unitary $U$ and final $Z$ measurements above. As the gadget measurements can be moved to the end, this whole process is a postselected CM circuit.

\end{proof}

\begin{corollary} \label{postmulti}
Any language in post-BQP can be decided with bounded error by a postselected CM circuit assisted by efficient classical computation. Thus if uniform families of CM circuits could be weakly classically simulated to within multiplicative error $1\leq c <\sqrt{2}$, then the polynomial hierarchy would collapse to its third level.
\end{corollary}
\begin{proof}
The first claim follows immediately from Theorem \ref{postselection}, and then the second follows from 
\cite{Bremner2011}.
\end{proof}


\subsection{Background for additive error case}

Before considering simulation of CM circuits up to additive error, we first outline a general framework and argument (following \cite{Aaronson2011,Bremner2016} but with some generalisation of context for our later purposes) that has been used in the literature (for example in \cite{Aaronson2011, Bremner2016,Fefferman,Morimae2017,Bouland2017,BFNV18}) to argue for hardness of classical simulation, up to additive error, of a variety of classes of quantum computational processes. 

Consider a given class $\cc = \{ C_\theta : \theta \in \Theta\}$ of quantum circuits parameterised by $\theta \in \Theta$, with each circuit also having its input state specified. We will generically denote the number of qubit lines of $C_\theta$ by $n$. Let the output be given by a measurement of all $n$ lines and let $p_\theta(x)$ with $x\in B_{n}$ denote the output probability distribution of $C_\theta$.

Introduce the following computational (sampling) task $\ct_\cc$ associated to the class $\cc$: for any given $\theta$, return $(\theta, y)$ where $y\in B_{n}$ has been sampled according to the output distribution $p_\theta$ of $C_\theta$. We will be interested in the complexity of simulating this task (and some approximate variants) as a function of $n$.

By an $\epsilon$-additive error simulation of the task $\ct_\cc$, we mean a process that given $\theta$, returns $(\theta, y')$ where $y'$ has been sampled according to a distribution $q_{\theta}$ on $B_{n}$ which is an $\epsilon$-additive approximation of the distribution $p_\theta$.

An alternative task (that neither a classical nor quantum computer is likely to be able to efficiently achieve) is to compute a value for $p_{\theta}(x)$ for given $\theta$ and $x$, up to a (suitably specified) multiplicative error. 
Indeed  for relevant classes that are studied in the literature, it can be shown that computing such approximations  is \#P hard in the worst-case. This task is of computational significance since for suitably chosen classes $\cc$ the probability values can be used to represent quantities that are of independent physical or mathematical interest.

Our aim is to argue for classical hardness of simulation of the sampling problem $\ct_\cc$ up to additive approximation. To do this we will need to conjecture that estimating the value of $p_{\theta}(x)$ up to (suitable) multiplicative approximation remains \#P hard not just in the worst-case, but in an average-case setting of the following kind.

For each class $\cc$ and number of lines $m$ introduce the set 
\[ \cd = \{ (\theta, x): C_\theta \mbox{ has $m$ lines and $x\in B_m$} \} . \]
For each $m$ we have a given probability measure $\pi$ on the set of $\theta$'s that occur in $\cd$, and let $\nu$ denote the uniform probability measure on $B_m$. Then $\pi \times \nu$ is the product measure on $\cd$. Finally, to the class $\cc$ we associate two constants: a measure size $0<f<1$ and an error tolerance $\eta$.
 
We introduce the following conjecture that we will refer to as Hardness$(\cc,\pi)$.\\[2mm]
{\bf Average-case hardness conjecture for $\cc$ with $\pi$}: \textit{
let ${\cal F} \subseteq \cd$ be any chosen subset of $\cd$ having $\pi\times \nu$ probability measure $f$. Then it is \#P hard to approximate the values $p_\theta(x)$ for all $(\theta,x)\in {\cal F}$ up to multiplicative error $\eta$}.\,$\Box$\\[2mm]
Note that if $\pi$ is the uniform measure too, then the subsets $\cf$ (for each $m$) will also be of fractional size $f$. But for nonuniform $\pi$'s there will be subsets of measure $f$ that have smaller fractional size than $f$ and asserting their \#P hardness is a stronger conjecture. The use of nonuniform distributions will also feature significantly in the anticoncentration property below.

As an example, in \cite{Bremner2016}  classes of IQP circuits $C$ are considered and conjectures 2 and 3 of \cite{Bremner2016} can be expressed as above, with $\pi$ being the uniform distribution, $f=1/24$ and $\eta= 1/4+o(1)$. In \cite{Bremner2017} the authors also consider the same classes of IQP circuits, but a nonuniform $\pi$ is used. This leads to a different average case hardness conjecture from those appearing in \cite{Bremner2016}.

The arguments below will use several complexity classes that we will loosely describe here in a way that suffices to express the hardness of simulation argument. For more complete descriptions see for example Ref\cite{ABbook}. $\rm{BPP}^{\rm{NP}}$ is the class of decision problems that can be solved by randomised classical polynomial time computations armed with an oracle for any problem in NP. $\rm{FBPP}^{\rm{NP}}$ is the same except that the outputs can be bit strings rather than just a single bit. $\rm{BPP}^{\rm{NP}}$ is in the third level of the tower of complexity classes known as the polynomial hierarchy PH. $\textrm{P}^{\#\textrm{P}}$ is the class of decision problems solvable in classical polynomial time, given access to an oracle for any \#P problem; and it is known (Toda's theorem) that  $\textrm{PH}\subseteq \textrm{P}^{\#\textrm{P}}$.

Now suppose that the sampling task $\ct_\cc$  can be solved up to additive error by a classical polynomial time algorithm $\mathcal{A}$. The first step is to show this ability to sample implies the existence of an  $\rm{FBPP}^{\rm{NP}}$ algorithm which, with use of $\ca$, can estimate $p_{\theta}(x)$ up to an additive error, for each $\theta$ and a constant fraction of  choices of $x$. After that an anticoncentration result will be used to convert the additive error into a multiplicative one, at least for a good measure of instances of $(\theta,x)$. The final step is to then invoke the average-case hardness conjecture for $\cc$: if our multiplicative approximation determination (computable in $\rm{FBPP}^{\rm{NP}}$) is ${\#\textrm{P}}$ hard then $\textrm{P}^{\#\textrm{P}}\subseteq \textrm{P}^{{\rm FBPP}^{\rm{NP}}} = {\rm BPP}^{\rm NP}$. The latter class is in the third level of PH and then by Toda's theorem, PH will collapse to its third level. However such a collapse is widely regarded as extremely implausible (similar to a collapse of NP to P), providing plausibility that the purported classical polynomial time algorithm $\mathcal{A}$ for solving $\ct_\cc$ up to additive error, cannot exist (if the average hardness conjecture is accepted).

\begin{lemma}\label{additive_lemma} (adapted from Lemma 4 of  \cite{Bremner2016})
Suppose there is a classical polynomial time algorithm $\mathcal{A}$ that simulates the sampling task $\ct_\cc$ up to additive error $\epsilon$. Then for any $0< \delta< 1$  there is an $\rm{FBPP}^{\rm{NP}}$ algorithm that, for each $\theta$,  approximates $p_{\theta}(x)$ up to additive error
\begin{equation} \label{additive}
\frac{p_{\theta}(x)}{\poly(n)}+(1+o(1))\cdot \frac{\epsilon}{2^{n} \delta}
\end{equation}
for at least a fraction $1-\delta$ of all $x\in B_{n}$. Thus for any probability measure $\pi$, the subset of $\cd$ to which eq. (\ref{additive}) applies, has $\pi\times\nu$ measure at least $1-\delta$ (since the measure of the full space of $\theta$'s is always unity).
\end{lemma}

This lemma is readily proved by following the argument of the proof of Lemma 4 in \cite{Bremner2016}, with minor notational modifications. 

To obtain a multiplicative error from this additive one, we require an  anticoncentration property of the following form.\\[2mm]
{\bf Anticoncentration property for $\cc$ with $\pi$}: \textit{there are constants $\alpha>0$ and $0\leq \beta\leq 1$ such that 
$p_\theta(x)\geq \alpha/2^{n}$ holds on a subset of $\cd$ of $\pi\times \nu$ measure at least $\beta$}.\, $\Box$\\[2mm]
In the literature a property of this form is proved for some classes $\cc$ (e.g. in \cite{Bremner2016,Bouland2017,Morimae2017,Bremner2017}) and conjectured to hold for others (e.g. in \cite{Aaronson2011}). Proofs of the property generally involve applying the Paley-Zygmund inequality to the probability measure $\pi\times\nu$.

Suppose now that the anticoncentration property holds for $\cc$. Then by choosing $\delta$ in Lemma \ref{additive_lemma} to be $\beta/2$ we guarantee an overlap $\Xi\subset \cd$ of probability measure at least $\beta/2$ on which the anticoncentration property $p_\theta(x)/\alpha\geq 1/2^{n}$  and the additive approximation  bound of eq. (\ref{additive}) both hold.

Then substituting  $p_\theta(x)/\alpha$ for $1/2^{n}$ in eq. (\ref{additive}) the approximation bound becomes 
\[
\frac{p_{\theta}(x)}{\poly(n)}+(1+o(1))\cdot \frac{2\epsilon}{\alpha\beta} p_\theta(x)
\]
giving a multiplicative approximation bound of size $\frac{2\epsilon}{\alpha\beta} + o(1)$ for $p_\theta(x)$, for a $\beta/2$ measure subset of $\cd$. 

Finally collecting all the above, we arrive at the following conclusion.
\begin{theorem} \label{sup}
Let $\cc$ be any class of quantum circuits with associated measure $\pi$ for which the anticoncentration property holds (with constants $\alpha$ and $\beta$). Suppose that the sampling task $\ct_\cc$ can be efficiently classically simulated up to additive error $\epsilon$.
Then if the average-case hardness conjecture holds with measure size $f=\beta/2$ and error tolerance $\eta=2\epsilon/(\alpha\beta)$, the polynomial hierarchy will collapse to its third level. 
\end{theorem}
For example in \cite{Bremner2016} we have $\epsilon =1/192$, and the anticoncentration property is shown to hold with uniform $\pi$, $\alpha=1/2$ and $\beta=1/12$. So to obtain collapse of PH we need the average-case hardness conjecture to be valid with error tolerance $\eta=2\epsilon/(\alpha\beta)=1/4$ and fraction $f=\beta/2 =1/24$.\\[1mm]

\subsection{Hardness of classical simulation of CM with additive error}

We now show that CM circuits cannot be classically efficiently simulated with additive error unless PH collapses, given average-case hardness conjectures. While CM circuits have been shown before \cite{Bouland2017,Pashayan2017} to have this property for one particular average-case-conjecture, here we show that actually a broad variety of such conjectures apply, such that if any one of them is proven, it implies the hardness of CM circuit simulation. Furthermore, in previous work, this hardness result for CM was shown by invoking the fact that Clifford gates form a 2-design \cite{Dankert2009} and that 2-designs anticoncentrate \cite{Hangleiter2017,Mann2017}, to give the needed anticoncentration  property. Here we follow a very different method, instead using the ability of CM circuits (via Therorem \ref{EGK}) to simulate any nonadaptive circuit. This allows CM circuits to simulate several other classes of circuits (not necessarily 2-designs) and inherit their average-case hardness conjecture as a basis for hardness of CM circuit simulation up to additive error.

Consider any class of unitary circuits $\cc = \{ C_\theta : \theta \in \Theta\}$ and associated measure $\pi$ on $\Theta$, for which a suitable anticoncentration property holds, and whose classical simulation up to additive error would imply collapse of PH if we assume Hardness$(\cc,\pi)$. Suppose that these circuits have been expressed as circuits of gates from the universal set of basic Clifford gates with $T$ and $T^\dagger$. We can use any choice of such a representation. Now consider the expanded class $\cc^T$ obtained by taking each circuit $C_\theta$ and replacing each $T$ and $T^\dagger$ gate by either $T$ or $T^\dagger$ in all combinations. If $C_\theta$ has $t$ $T$ and $T^\dagger$ gates then it will give rise to $2^t$ circuits in $\cc^T$, and these can be labelled by $(\theta,\tau)$ where $\tau$ is a $t$-bit string indicating the choices of $T$ and $T^\dagger$. Accordingly, we write $\cc^T=\{ C_{\theta,\tau} : \theta\in \Theta,\, \tau\in B_t \}$.

$\cc^T$ is exactly the class of circuits we obtain if we implement the circuits $\C_{\theta}$ using $T$ gadgets for each $T$ and $T^\dagger$ gate, but omit all the adaptive $S$ gate corrections that are normally specified by the $T$-gadget measurement outcomes. Denote that non-adaptive circuit by $U_\theta$ with outputs $(x,\tau)$ where $\tau\in B_t$ is the string of gadget measurement outcomes and $x$ arises from the output lines from $C_\theta$.
Each of the $2^t$ possibilities for $\tau$ will occur with equal probability. Note that the circuits $U_\theta$ are unitary Clifford circuits (having only final $Z$ measurements). Indeed the measurement within any (generally intermediate) $T$-gadget can now be moved to the end of the circuit as that line is not acted on again, and the measurement outcome is not used in any adaptations. Because these circuits are unitary Clifford circuits, they can be simulated by CM circuits using Theorem \ref{EGK} (ii). Denote the associated CM circuit (with input state $\ket{A}^{\otimes t}$) by $V_\theta$. Finally let $p_\theta(x)$, $p_{\theta, \tau}(x)$ and $u_\theta(x,\tau)$ (with $x\in B_n$, $\tau\in B_t$) denote the output probabilities for the circuits $C_\theta$, $C_{\theta,\tau}$ and $U_\theta$ respectively.  
  
  Note that for each $\theta$ there is a $\tau_0=\tau_0(\theta)$ for which $p_{\theta,\tau_0}(x) =p_\theta(x)$, viz. $\tau_0$ just specifies the $T$ and $T^\dagger$ choices that actually occur in $C_\theta$. Furthermore, since each $\tau$ arises in the output of $U_\theta$ with equal probability $1/2^t$, the relationship between  $C_{\theta,\tau}$ and $U_\theta$ gives (via conditional probabilities):
    \begin{equation}\label{probrel}
   p_{\theta, \tau}(x) = u_\theta(x,\tau)\, 2^t.
   \end{equation}
   Finally in addition to distribution $\pi$ on the $\theta$'s, let $\nu$ and $\nu'$ denote the uniform distribution on the $x$'s and $\tau$'s respectively. Let $\pr_{\pi\times\nu\times\nu'}(\theta,x,\tau)$ denote the probability of $(\theta,x,\tau)$ in the product distribution $\pi\times\nu\times\nu'$, and similarly for $\pr_{\pi\times\nu'}(\theta,\tau)$, $\pr_\pi(\theta)$ etc.
   
We will show that, for some classes $\cc$ of circuits already proved to have the additive simulation hardness property of Theorem \ref{sup} (subject to an associated Hardness$(\cc,\pi)$ conjecture), that $\cc^T$ contains no new circuits that were not already present in $\cc$. Thus the labels $(\theta,\tau)$ will label the circuits of $\cc$ with generally high redundancy, and we write $\cc^T=\cc$ in this situation.
Since such circuits can be simulated by CM circuits, classical simulation of CM circuits up to additive error can then imply collapse of PH, subject to the conjecture Hardness$(\cc,\pi)$ of the class $\cc$, as will be formalised in the Theorem below.

Suppose now that $\cc=\cc^T$. Then for each $(\theta,\tau)$ there is $\tilde{\theta}=\tilde{\theta}(\theta,\tau)$ with $C_{\theta,\tau}$ being $C_{\tilde{\theta}}$ so
   \[ p_{\theta, \tau}(x) =p_{\tilde{\theta}}(x). \]
   We will also require the following $\theta$-sampling relation: the $C_\theta$ circuits occurring multiply in $\cc^T$, occur with the same probability in $\cc^T$ (wrt distribution $\pi\times\nu'$)  as they did in $\cc$ (wrt distribution $\pi$): \begin{equation}\label{probcond}
   \sum_{(\theta,\tau):\tilde{\theta}(\theta,\tau)=\theta_0} \pr_{\pi\times\nu'}(\theta,\tau)=\pr_\pi(\theta_0).
   \end{equation}
   
  \begin{theorem} \label{main1}
Consider any class of circuits $\cc$ with associated distribution $\pi$ for which the following hold:\\
(i) the anticoncentration property (with parameters $\alpha$ and $\beta$);\\
(ii) $\cc=\cc^T$ and the $\theta$-sampling relation eq. (\ref{probcond}).\\
Then if every CM circuit can be efficiently classically simulated to additive error $\epsilon$, the average-case hardness conjecture for $(\cc,\pi)$ with parameters $f=\beta/2$ and $\eta=2\epsilon/(\alpha\beta)$ will imply that PH collapses.
\end{theorem}

\begin{proof} 
We use the notations and definitions introduced above. Since $U_\theta$ can be simulated by a CM circuit, if every CM circuit can be efficiently classically simulated to additive error $\epsilon$, then so can the distribution $u_\theta(x,\tau)$. So by Lemma \ref{additive_lemma} applied in $(\theta,\tau,x)$ space, there is a $(1-\beta/2)$ sized subset in $\pi\times\nu'\times\nu$ measure where an $\rm{FBPP}^{\rm{NP}}$ algorithm can calculate an additive approximation to $u_\theta(x,\tau)$ with additive error bound of 
\begin{equation}\label{add1}
\frac{u_\theta(x,\tau)}{\poly(n+t)}+(1+o(1))\cdot \frac{2\epsilon}{2^{n+t} \beta}
\end{equation}
(since we have $n+t$ lines now). 

Next we will want a measure $\beta$ subset of $(\theta,\tau,x)$'s on which the anticoncentration property $u_\theta(\tau,x)\geq \alpha/2^{n+t}$ holds. By $(\cc,\pi)$ anticoncentration, there is a measure $\beta$ subset of $(\theta,x)$'s with $p_\theta(x)\geq \alpha/2^n$. So by the $\theta$-sampling relation eq. (\ref{probcond}) and eq. (\ref{probrel}) there is a measure $\beta$ subset of $(\theta,\tau,x)$'s with 
\begin{equation}\label{conc}
u_\theta(x,\tau)=\frac{p_{\theta,\tau}(x)}{2^t}\geq \frac{\alpha}{2^{n+t}}
\end{equation}
(noting that for any $x$, $\pr_{\pi\times\nu}(\theta,x)=\pr_\pi(\theta)/2^n$). Combining eqs. (\ref{conc}) and (\ref{add1}) we get a measure $\beta/2$ subset of $(\theta,\tau,x)$'s on which $u_\theta(x,\tau)$ can be calculated by an 
$\rm{FBPP}^{\rm{NP}}$ algorithm to multiplicative approximation $2\epsilon/(\alpha\beta)+o(1)$, and this also applies to $p_{\theta,\tau}(x) = u_\theta(x,\tau)2^t$ (as multiplicative approximations are invariant under scale changes).

Finally we want to map this back to $(\theta,x)$ space. Note that for any $x$
\[ \pr_{\pi\times\nu'\times\nu}(\theta,\tau,x)=\frac{1}{2^n}\pr_{\pi\times\nu'}(\theta,\tau)\leq 
\frac{1}{2^n}\pr_{\pi}(\tilde{\theta}(\theta,\tau))=\pr_{\pi\times\nu}(\tilde{\theta},x)
\]
(where the inequality follows from eq. (\ref{probcond})). Hence the map $(\theta,\tau,x)\mapsto (\tilde{\theta}(\theta,\tau),x)$ gives a subset of $(\theta,x)$'s of measure $\geq \beta/2$ on which $p_\theta(x)$ can be calculated to multiplicative approximation $2\epsilon/(\alpha\beta)+o(1)$ by an $\rm{FBPP}^{\rm{NP}}$ algorithm. Hence the average-case hardness conjecture for $(\cc,\pi)$  implies that PH collapses to its third level.
\end{proof}

Examples of circuit classes in the literature for which a suitable anticoncentration property holds, $\cc=\cc^T$ and the $\theta$-sampling relation eq. (\ref{probcond}) holds, include the following.\\[2mm]
\textbf{IQP circuits associated with the Ising model  \cite{Bremner2016}}\\
This is the class of circuits $\cc$ having input $\ket{0}^{\otimes n}$ acted on by $H^{\otimes n} U H^{\otimes n}$, where $U$ is  unitary and chosen in the following way: apply $T^{v_i}$ to each qubit line $i$, and $CS^{w_{ij}}$ to each pair of qubits $i,j$, where $v_i$ and $w_{ij}$ (all collectively comprising the label $\theta$) are chosen in all possible combinations from $\{0,...,7\}$ and $\{0,...,3\}$ respectively, and $CS$ is the controlled-$S$ gate. Furthermore the $CS$ gate is implemented in terms of Clifford+T+T$^\dagger$ gates using the gadget of Figure 3. The distribution $\pi$ is the uniform distribution.
\begin{figure}[h!b] \label{CS}
    \begin{center}
    \includegraphics[trim=0 0 0 350,clip,width=\textwidth]{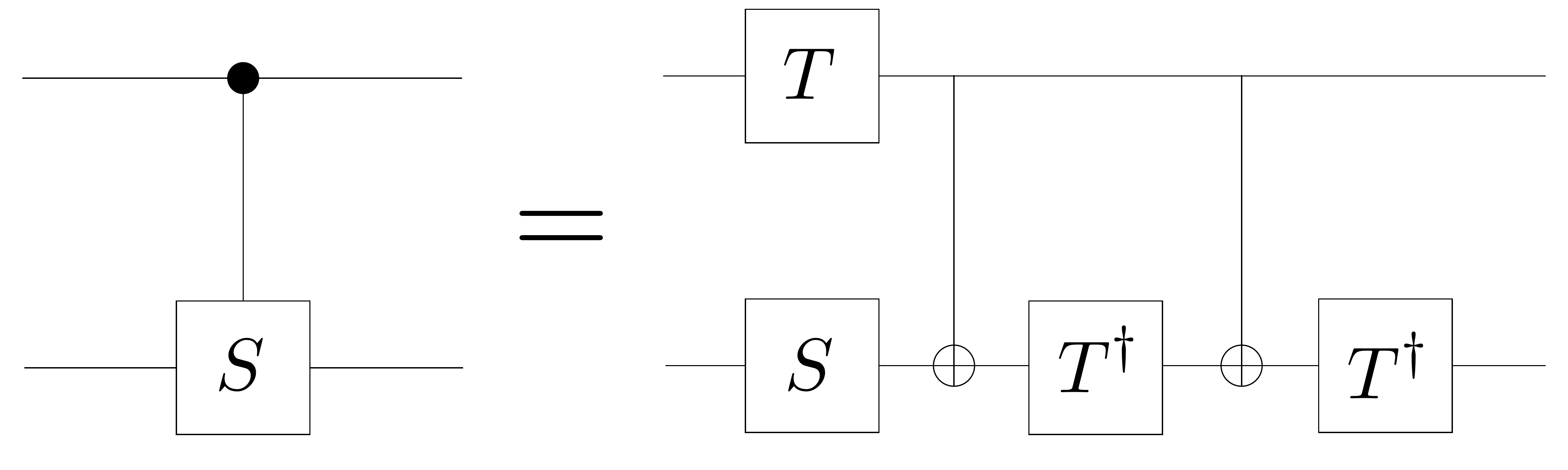}
    \caption{Decomposing the controlled-$S$ gate into Clifford+$T$+$T^\dagger$ gates.}
    \end{center}
   
\end{figure}

To see that $\cc=\cc^T$ note first that if any initial $T$ or $T^\dagger$ gates are changed (to the other choice), the resulting circuit is clearly still a circuit in the original set. However, there are also $T$ and $T^\dagger$ gates within the $CS$ gadget of Figure 3 to consider.
If the $T$ or $T^\dagger$ gates at either end are changed, this can be corrected by applying further $T$ gates. If the middle $T^\dagger$ gate is swapped, the result is $CS \,(T\otimes T)$. So in each of these cases, the resulting circuit is still from the original set. The $\theta$-sampling relation eq. (\ref{probcond}) holds because for each $\theta$ there is a $\tau_0=\tau_0(\theta)$ with $\tilde{\theta}(\theta,\tau_0)=\theta$ and the fact that for any fixed $\tau'$ (and varying $\theta$) the mapping  $(\theta,\tau_0(\theta))\mapsto (\theta,\tau_0\oplus \tau')$ is one-to-one on the underlying $\tilde{\theta}$'s (with $\oplus$ being addition of $t$-bits strings at each entry).\, $\Box$\\[2mm]
\textbf{Sparse IQP circuits \cite{Bremner2017}}\\
This class is the same as the above (so $\cc=\cc^T$) but with a different distribution $\pi$.  Specifically, having chosen each $v_i$ and $w_{ij}$ uniformly, each $CS^{w_{ij}}$ gate is applied only with some probability p, while each  $T^{v_i}$ is applied as in the above case. This amounts to $w_{ij}=0$ being chosen with probability $\frac{1}{4}+\frac{3}{4}(1-p)$ and other $w_{ij}$'s with probability $p/4$ (and $v_i$'s chosen uniformly as before). Also as before when a $T$ gate inside of $CS$ is swapped, it always becomes $CS$ with some extra $T$ gates. The $\theta$-sampling relation eq. (\ref{probcond}) holds since reassigning $T$ and $T^\dagger$ gates always preserves the number of two qubit gates in the circuit.\, $\Box$ \\[2mm]
\textbf{Random Circuit Sampling \cite{Boxio2017}}\\
Another class of circuits was put forward by the Google/UCSB team, and called random circuit sampling. The gates used in these circuits are from $\{CZ,X^{1/2},Y^{1/2}, T\}$. In \cite{Hangleiter2017} it is shown that circuits from this set anticoncentrate if they are chosen as follows: 
let $G=\{CZ,X^{1/2},X^{-1/2},Y^{1/2},Y^{-1/2}, T,T^\dagger\}$ (i.e. the previous set closed under inverses). In each time step either $U_{1,2}\otimes U_{3,4}\otimes ... \otimes U_{n-1,n}$ or $U_{2,3}\otimes U_{4,5}\otimes ... \otimes U_{n-2,n-1}$ is applied, for all possible choices of $U_{j,j+1}$ from $G$ (with 1-qubit gates $U$ appearing as $I\otimes U$ or $U\otimes I$). Finally all $n$ lines are measured in the computational basis. The distribution $\pi$ over $\cc$ is the uniform distribution.
All gates in $G$ besides $T$ and $T^\dagger$ are Clifford, so reassigning $T$ and $T^\dagger$ gates clearly results in circuits from the same class i.e. $\cc=\cc^T$, and a uniform distribution for $\pi$ satisfies eq. (\ref{probcond}).

In  \cite{BFNV18} it is shown that Random Circuit Sampling has a property similar to the required average-case hardness result viz. that the conjecture holds if the task is to compute $p_{\theta}(x)$ exactly. This is known to be \#P hard, even for the average case. Boson sampling  \cite{Aaronson2011} is the only other class where this is kind of result has been proved. Although referring to exact calculation, this can nevertheless be viewed as providing evidence that the necessary average-case hardness conjecture (involving approximate computation, up to multiplicative error) may hold. \, $\Box$

CM circuits simulating any one of these three classes inherit the hardness of the original circuits. If average-case hardness is shown for any of them then it implies the same is true for CM circuits and therefore that CM cannot be efficiently classically simulated up to additive error. This result is a natural consequence of the Extended Gottesman--Knill theorem that shows how CM circuits can simulate other types of quantum computations.


For other classes of circuits we generally have $\cc\neq\cc^T$ i.e. $\cc^T$ contains circuits that were not already present in $\cc$. However, if $\cc^T$ also has a suitable anticoncentration property, then up to an average-case hardness conjecture, PH will collapses if $\cc^T$ circuits can be classically simulated to additive error. Note that if $\cc$ has a worst-case hardness result (as is generally the case for classes considered), then so does $\cc^T$ since its circuits always form a superset of $\cc$. This provides evidence for a suitably analogous average-case conjecture for $\cc^T$. Hence, in the case that $\cc^T$ also anticoncentrates, it is also likely to be hard to classically simulate. For any $\cc$, the circuits in $\cc^T$ can always be simulated by CM circuits (in the sense above, used in Theorem \ref{main1}, taking the uniform distribution over the $\tau$'s as above) and we obtain the following result.
\begin{theorem} \label{main2}
Suppose that $\cc^T$ (arising from $(\cc,\pi)$ as described above) satisfies an anticoncentration property with constants $\alpha$ and $\beta$. Then if every CM circuit can be efficiently classically simulated to additive error $\epsilon$, PH will collapse to the third level if we assume an average hardness conjecture for $\cc^T$ with parameters $f=\beta/2$ and $\eta=2\epsilon/(\alpha\beta)$, extending the corresponding conjecture for $\cc$. Furthermore, if $\cc$ had the worst-case hardness property, then so does $\cc^T$. 
\end{theorem}

One example of circuits for which $\cc \subsetneq \cc^T$ and $\cc^T$ also anticoncentrates, is the class of \textbf{Conjugated Clifford circuits} introduced in \cite{Bouland2017}. Here we have circuits of the form $V^{\otimes n \dagger}  U V^{\otimes n}$, where $V$ is any fixed 1-qubit gate and $U$ is any Clifford circuit (so we get a class for each choice of $V$), and $\pi$ is the uniform distribution. The representation of $V$ in terms of Clifford+T+T$^\dagger$ gates generally contains $T$ and $T^\dagger$ gates, and when these are reassigned in all combinations in $V^{\otimes n}$, the result is no longer necessarily a gate of the form $W^{\otimes n}$ i.e. the gates applied on different lines will generally be different, and the $n$-qubit gate on one end will also not necessarily be the inverse of the one on the other end. Hence $\cc \subsetneq \cc^T$. However, this new class of circuits does anticoncentrate. This follows from the original anticoncentration proof in Ref \cite{Bouland2017} (Lemma 4.3 there) which still applies for arbitrary $n$-qubit gates replacing $V^{\otimes n}$ and $V^{\otimes n\dagger}$ on the ends.

We expect there to be other classes to which Theorem \ref{main2} can be applied, providing further corresponding average hardness conjectures which suffice to make CM circuits hard to classically simulate up to additive error. That is because a common strategy for proving that a class of circuits anticoncentrates is to show that it is an $\epsilon$-approximate 2 design and then use the result \cite{Hangleiter2017,Mann2017}, that such 2-designs have the anticoncentration property. In this vein the following conjecture if true, would be a useful result.

\begin{conjecture}\label{conjdesign}
Suppose $\cc$ with $\pi$ is an $\epsilon$-approximate 2 design. Then $\cc^T$ with $\pi\times \nu$ is also an approximate 2 design. \end{conjecture}
The circuit class $\cc^T$ depends on the choice of representation of circuits in $\cc$ in terms Clifford+T+T$^\dagger$ gates. 
If Conjecture \ref{conjdesign} were to hold for just one choice of such a representation for $(\cc,\pi)$  that is an $\epsilon$-approximate 2 design, then the conclusions of Theorem \ref{main2} will apply.


\noindent\textbf{Acknowledgements.} We thank M. Bremner and A. Montanaro for helpful discussions and clarifications. We thank Ryuhei Mori and an anonymous referee for pointing out approach (b) in Theorem 5.1 to us.
We acknowledge Mem Fox for suggesting terminology.\\

\noindent\textbf{Ethics statement.} This work did not involve any issues of ethics.\\

\noindent\textbf{Data accessibility statement.} This work does not have any experimental data.\\

\noindent\textbf{Competing interests.} There are no competing interests for this paper.\\

\noindent\textbf{Funding statement.} We acknowledge support from the QuantERA ERA-NET Cofund in Quantum Technologies implemented within the European Union's Horizon 2020 Programme (QuantAlgo project), and administered through the EPSRC grant EP/R043957/1. MY is supported by the Australia Cambridge Bragg Scholarship scheme, and SS by the Leverhulme Early Career Fellowship scheme.\\

\noindent\textbf{Authors' contributions.} All authors in collaborative work made substantial contributions to conception and drafting of this work. All authors gave final approval for publication.

\bibliographystyle{abbrv}

\bibliography{Clifford_Magic_bibtex}

\end{document}